\documentclass[aip,jmp,amsmath,amssymb,reprint,onecolumn,author-year]{revtex4-1}

\usepackage{graphicx}
\usepackage{dcolumn}
\usepackage{bm}

\usepackage{enumerate}

\usepackage{ramsthm,amsfonts,bbm,mathrsfs} 

\newtheorem{definition}{Definition}
\newtheorem{lemma}{Lemma}
\newtheorem{corollary}{Corollary}
\newtheorem{theorem}{Theorem}

\newtheorem{example}{Example}

\DeclareMathOperator{\supp}{supp}
\DeclareMathOperator{\Id}{Id}

\def\todo#1{\textcolor{Rot}{TODO\ifthenelse{\equal{#1}{}}{}{: #1}}}

\begin{document}



\preprint{AIP/123-QED}

\title[Rainer M\"{u}hlhoff -- Cauchy Problem, Green's Functions and Algebraic Quantization]{Cauchy
Problem and Green's Functions for First Order Differential Operators and
Algebraic Quantization}

\author{Rainer M\"{u}hlhoff}%
 \email{rainer.muehlhoff@tu-berlin.de}%
 \affiliation{Berlin Mathematical School, TU Berlin, Mathematisches Institut,
 MA 2-2, Stra\ss e des 17. Juni 136, 10623 Berlin, Germany}%
 \altaffiliation{now at Technische Universit\"{a}t Berlin, Room EB 238-240,
 Sekr. EB 4-1, Stra\ss e des 17. Juni 145, 10623 Berlin, Germany}%

\date{January 22, 2010}

\begin{abstract}
	\textbf{Abstract.} 
	Existence and uniqueness of advanced and retarded fundamental
	solutions (Green's functions) and of global solutions to the Cauchy
	problem is proved for a
	general class of first order linear differential operators on vector
	bundles over globally hyperbolic Lorentzian manifolds.	
	This is a core ingredient to \mbox{CAR-/}CCR-algebraic
	constructions of quantum field theories on curved spacetimes,
	particularly for higher spin field equations. 

	\bigskip
	\noindent This article appeared in J. Math. Phys. 52, 022303 (2011);
	doi:10.1063/1.3530846 and may be found at
	\url{http://link.aip.org/link/?JMP/52/022303}.
\end{abstract}

\maketitle

\section{Motivation and Physical Goal}

In the setting of general relativity, various types of physical fields
(and, in a suitable sense, also quantized fields) can be described
as solutions $\Phi$ to a differential equation
\[
P \Phi = 0\,, \quad\Phi \in \Gamma(\mathcal{E}) 
\,.\]
Here, $\mathcal{E}$ is a real or complex vector bundle over the Lorentzian
spacetime manifold $M$ with metric $g$, and $P$ is a linear differential
operator while $\Gamma(\mathcal{E})$ denotes the space of smooth sections in
$\mathcal{E}$.  For the scalar case where $P$ is the Klein-Gordon operator and
$\mathcal{E}$ is the complex line bundle on $M$ as well as for the spin
$\frac{1}{2}$ case where $P$ is the Dirac operator and $\mathcal{E}$ the
bundle of Dirac spinors on $M$, \cite{Dimock1980,Dimock1982} proposed
algebraic quantum field theory constructions based on certain $C^*$-algebras
(CAR-/CCR-algebras) of local observables, derived from the field operators
(for an outline of such a construction cf.\ sec.\
\ref{sec:quantisation-outline}).  It was shown that these quantum field
theories are in fact locally covariant, cf.\
\cite{BrunettiFredenhagenVerch2003}, \cite{Sanders2010}.  Core ingredients to
both the constructions and the general covariance result are the
well-posedness of the Cauchy problem and the existence of unique advanced and
retarded fundamental solutions (Green's functions) for the differential
operator $P$, which are both, in the end, consequences of global hyperbolicity
of $(M,g)$.

It is of course desirable to have such a quantum field theory for every type
of quantum fields (particularly for \textit{fields of higher spin}) or at
least to have established for which fields such a construction may be
possible.  This is why in this work, for a quite general class of first order
differential operators on globally hyperbolic spacetimes we present results on
existence of advanced and retarded Green's functions and well-posedness of the
Cauchy problem, which are crucial requirements for the CAR-/CCR-algebraic
quantum field theory construction as indicated above. 

While many of the classic references (like \onlinecite{Dimock1980,Dimock1982})
still had to refer back to the famous unpublished lecture notes by
\cite{Leray1953} when dealing with these questions for the scalar and spin
$\frac{1}{2}$ cases, meanwhile we have a powerful theory of \textit{second
order} normally hyperbolic differential operators on globally hyperbolic
spacetimes at hands, presented in modern mathematical language by B\"{a}r,
Ginoux, Pf\"{a}ffle (\citeyear{BaerGinPfae}). We shall make use of this work
to derive our results on the first order case. 

The starting point of our work was an investigation of the generalized Dirac
equations for higher spin fields as proposed by \cite{Buchdahl1982a}, for
which so far there existed only \textit{local} results on the well-posedness
of the Cauchy problem (cf.\ \onlinecite{Wuensch1985}).  It turned out that for
wide parts of our considerations (which may be found in full detail in
\onlinecite{Dipl}), committing to a special case of differential operator was
not necessary. The present paper presents the first and general part of our
work, while in a forthcoming publication we are planning to address
applications to Buchdahl's and other higher spin field equations.

\section{Introduction and Notation}

Throughout this article, let $(M,g)$ be an $n$-dimensional smooth, time
oriented, connected, globally hyperbolic Lorentzian manifold of signature
$(+--\ldots)$ and let $\mathcal{E}$ be a smooth, finite rank vector bundle on
$M$.  We denote the space of smooth sections of $\mathcal{E}$ by
$\Gamma(\mathcal{E})$ and the space of compactly supported smooth sections by
$\Gamma_0(\mathcal{E})$. 

A second order linear differential operator $L\colon \Gamma(\mathcal{E}) \to
\Gamma(\mathcal{E})$ is called \textbf{normally hyperbolic,} if its principal
symbol $\sigma_L$ ``is given by the metric'': 
\[
\forall x\in M\,\forall \xi\in T_x^*M\colon \sigma_L(\xi) = g(\xi,\xi)\,
\Id_{\mathcal{E}_x}
\,,\]
where $g$ also denotes the inverse metric of $g$ on $T^*M$. (For a definition
of the principal symbol $\sigma_L$ cf.\ e.\,g.\ \onlinecite{BaerGinPfae}.)
Equivalently this means that $L$ can locally be written as 
\[
L\Phi = g^{\mu\nu}\frac{\partial^2}{\partial x^\mu\partial x^\nu}\Phi +
\text{lower order derivatives of $\Phi$}
\,,\quad \Phi\in\Gamma(\mathcal{E}) 
\,.
\]
For normally hyperbolic (hence, second order) operators, the global Cauchy
problem and existence of advanced and retarded Green's functions 
were extensively treated in \cite{BaerGinPfae}. 
As this article is about \textit{first order} operators, we shall now define
a property which will play a role in our theory analogous to normal
hyperbolicity for the second order case: 

\begin{definition} \label{def:pre-norm-hyp}
	We call a first order linear differential operator $P\colon
	\Gamma(\mathcal{E})\to \Gamma(\mathcal{E})$ \textbf{prenormally
	hyperbolic,} if there is another first order linear differential
	operator $Q\colon \Gamma(\mathcal{E})\to \Gamma(\mathcal{E})$, such
	that $PQ$ is normally hyperbolic. 
	We call such a pair of $P$ and $Q$ a \textbf{complementary pair of
	prenormally hyperbolic first order operators.}
\end{definition} 

\noindent As a first useful property of prenormally hyperbolic differential
operators, we shall prove: 

\begin{lemma} \label{lem:invertibility-sigma-P}
	Let $P\colon\Gamma(\mathcal{E})\to \Gamma(\mathcal{E})$ be a
	prenormally hyperbolic first order differential operator. 
	Denote by $\sigma_P$ its principal symbol. Then for every $x\in M$
	and for every $\xi\in T^*_xM$ with $g_x(\xi,\xi)\ne 0$, the
	endomorphism $\sigma_P(\xi)\colon \mathcal{E}_x\to \mathcal{E}_x$ is
	invertible. 
\end{lemma}

\begin{proof}
	Let $Q\colon\Gamma(\mathcal{E})\to \Gamma(\mathcal{E})$ be another
	prenormally hyperbolic first order differential operator such that
	$Q,P$ form a complementary pair. Denote by $\sigma_Q$ the principal
	symbol of $Q$. 
	For the principal symbol $\sigma_{QP}$ of $QP$ we find at $x$: 
	$\sigma_{QP}(\xi) = \sigma_Q(\xi)\circ \sigma_P(\xi) 
	= g(\xi,\xi)\,\Id_{\mathcal{E}_x}$.
	Hence, $\sigma_P(\xi)$ is an automorphism if $g(\xi,\xi)\ne 0$. 
\end{proof}

\noindent It turns out that in definition \ref{def:pre-norm-hyp}, instead of
$PQ$ being normally hyperbolic we could equivalently demand that $QP$ be
normally hyperbolic: 

\begin{lemma}
	Let $P,Q\colon\Gamma(\mathcal{E})\to \Gamma(\mathcal{E})$ be first order
	linear differential operators. Then if $PQ$ is normally hyperbolic, 
	$QP$ is normally hyperbolic, too. 
\end{lemma}

\begin{proof}
	$PQ$ being normally hyperbolic means at every point $x\in M$: 
	\[
	\forall \xi\in T^*_xM\colon 
	\sigma_P(\xi) \circ \sigma_Q(\xi) =
	g_x(\xi,\xi)\,\Id_{\mathcal{E}_x}
	\,.\]
	First, let the co-vector $\xi$ be non-null. Then by lemma
	\ref{lem:invertibility-sigma-P}, $\sigma_P(\xi)$ and $\sigma_Q(\xi)$
	are invertible and thus elements of $\mathrm{GL}(\mathcal{E}_x)$. 
	As in groups right inverses are also left inverses, 
	we can deduce that 
	\[
	\sigma_Q(\xi) \circ \sigma_P(\xi) =
	g_x(\xi,\xi)\,\Id_{\mathcal{E}_x}
	\,.\] 
	Finally, as both sides of this equation depend continuously on $\xi$,
	it also holds for null co-vectors. 
	Hence we conclude that $QP$ is normally hyperbolic. 
\end{proof}

\begin{example}
	Let $M$ have a spin structure, let $\mathcal{D}M$ be the spinor bundle
	on $M$, let $D\colon \Gamma(\mathcal{D}M) \to \Gamma(\mathcal{D}M)$ be
	the Dirac operator\footnote{%
		For general terminology, cf.\ e.\,g.\ \cite{BerGetVer}
		or \cite{FewsterVerch2002}.%
	} and let $A\colon \Gamma(\mathcal{D}M) \to
	\Gamma(\mathcal{D}M)$ be any operator of order 0 (i.\,e.\ a linear map
	on the fiber). 
	
	Then $D+A$ is prenormally hyperbolic since using
	the Lichnerowicz formula we find: 
	\[
	(D+A)(D-A) = \Box + \frac{\mathrm{scal}}{4} - [D, A] - A^2
	\,,\] 
	where $\mathrm{scal}$ denotes the scalar curvature of the metric $g$
	on $M$ and $\Box$ denotes the Laplacian with respect to the covariant
	derivative on $\mathcal{D}M$ induced by the Levi-Civita covariant
	derivative on $TM$ (d'Alembertian). $\Box$ is the only
	second order
	term and it is normally hyperbolic. Hence, $(D+A)(D-A)$ is normally
	hyperbolic. 
\end{example}

\section{Existence of Green's functions}

As \cite{Dimock1982} did for the Dirac special case, we shall now
establish existence and uniqueness of advanced and retarded Green's operators
(in Dimock's terminology: fundamental solutions) for our general class of
prenormally hyperbolic operators. We do this by generalizing Dimock's proof,
facilitating results from \cite{BaerGinPfae}. 

\begin{definition}
	\label{def:greens-operators}
	Let $P\colon \Gamma(\mathcal{E})\to \Gamma(\mathcal{E})$
	be a first order linear differential operator on sections of
	$\mathcal{E}$.  
	Linear maps $G_\pm\colon \Gamma_0(\mathcal{E})\to \Gamma(\mathcal{E})$
	are called \textbf{advanced ($-$) resp.\ retarded ($+$) Green's
	functions for $P$}, if 
	\begin{enumerate}[(i)]
		\item $\displaystyle P\circ G_\pm =
			\Id_{\Gamma_0(\mathcal{E})}$\,, 
			
		\item $\displaystyle G_\pm\circ P =
			\Id_{\Gamma_0(\mathcal{E})}$\,, 

		\item $\displaystyle \forall \varphi\in
			\Gamma_0(\mathcal{E})\colon
			\supp(G_\pm\varphi)\subseteq J_\pm(\supp(\varphi))$
			\,.
	\end{enumerate}
\end{definition}

\noindent Here, for $A\subseteq M$, $J_+(A)$ resp.\ $J_-(A)$ denotes the
\textbf{causal future resp.\ past} of the set $A$, i.\,e.\ the set of points
in $M$, which are either points in $A$ or which can be reached from a point in
$A$ by a future resp.\ past directed, piecewise $C^1$ curve. 

\begin{theorem}
	\label{thm:greens-operators}
	Let $P\colon \Gamma(\mathcal{E})\to \Gamma(\mathcal{E})$ be a 
	prenormally hyperbolic first order linear differential operator on
	sections of $\mathcal{E}$. 
	Then there exist unique advanced and retarded
	Green's functions $S_\pm\colon \Gamma_0(\mathcal{E})\to
	\Gamma(\mathcal{E})$ for $P$. 
\end{theorem}

\begin{proof} 
	Choose $Q\colon \Gamma(\mathcal{E})\to \Gamma(\mathcal{E})$ such that 
	$P,Q$ form a complementary pair of prenormally hyperbolic first
	order operators (i.\,e.\ $PQ$ and $QP$ are normally hyperbolic). 
	Let $\mathcal{E}^*$ denote the dual bundle of $\mathcal{E}$ and let 
	$P^*,Q^*\colon \Gamma(\mathcal{E}^*)\to\Gamma(\mathcal{E}^*)$ denote
	the formal adjoints\footnote{%
		$P^*$ is the formal adjoint of $P$, if 
		$\forall \varphi\in\Gamma_0(\mathcal{E})\,~\forall
		\psi\in\Gamma_0(\mathcal{E}^*)\colon \int_M\psi(P\varphi) =
		\int_M(P^*\psi)\varphi$, where integration is with respect to
		the volume density induced by the Lorentz metric on $M$.%
	} of $P$ resp.\ $Q$. 
	
	As $QP$ and $PQ$ are normally hyperbolic, so are
	$(QP)^* = P^*Q^*$ and $(PQ)^* = Q^*P^*$\,\footnote{%
		The formal adjoint of a normally hyperbolic operator is again
		normally hyperbolic \cite[7.2.4]{Dipl}.%
	}. This enables us to apply 
	\onlinecite[corollary 3.4.3]{BaerGinPfae}, which says that there exist
	unique advanced and retarded Green's functions $G_\pm\colon
	\Gamma_0(\mathcal{E})\to \Gamma(\mathcal{E})$ for $PQ$
	and $G'_\pm\colon \Gamma_0(\mathcal{E}^*)\to
	\Gamma(\mathcal{E}^*)$ for $P^*Q^*$. 

	Define $S_\pm := QG_\pm$ and $S'_\pm := Q^*G'_\pm$. We will show that 
	$S_\pm$ are advanced/retarded Green's functions for $P$ 
	and that they are unique. 
	As a byproduct we will find that $S'_\pm$ are advanced/retarded
	Green's functions for $P^*$. We have to check all three conditions
	from definition \ref{def:greens-operators}: 

	\begin{enumerate}[\quad (i)]
		\item $PS_\pm|_{\Gamma_0(\mathcal{E})} 
			= PQG_\pm|_{\Gamma_0(\mathcal{E})}
			=\Id_{\Gamma_0(\mathcal{E})}$ and
			$P^*S'_\pm|_{\Gamma_0(\mathcal{E^*})}
			=P^*Q^*G'_\pm|_{\Gamma_0(\mathcal{E^*})}
			=\Id_{\Gamma_0(\mathcal{E}^*)}$ as $G_\pm$ resp.\
			$G'_\pm$ are Green's functions for $PQ$ resp.\
			$P^*Q^*$. 
			
		\item[(iii)] 
			For all $\varphi\in\Gamma_0(\mathcal{E})$ we have 
			\[
			\supp S_\pm\varphi 
			= \supp QG_\pm\varphi 
			\subseteq \supp G_\pm \varphi
			\subseteq J_\pm (\supp \varphi)
			\,,
			\]
			where the second $\subseteq$ makes use of the support
			property of $G_\pm$. 
			The argument for $S'_\pm$ is the same. 

		\item[(ii)] 
			We have to show $S_\pm
			P|_{\Gamma_0(\mathcal{E})}=
			\Id_{\Gamma_0(\mathcal{E})}$
			and $S'_\pm P^*|_{\Gamma_0(\mathcal{E}^*)}=
			\Id_{\Gamma_0(\mathcal{E}^*)}$. 
			Fix some arbitrary $\psi\in \Gamma_0(\mathcal{E}^*)$
			and $f\in \Gamma_0(\mathcal{E})$. Writing 
			$\langle \psi, f\rangle  := \int_M \psi(f)$,
			we find 
			\[
			\langle S'_\mp \psi, f\rangle  
			= \langle S'_\mp\psi, PS_\pm f\rangle \nonumber 
			\overset{(*)}{=} \langle P^*S'_\mp\psi, S_\pm
			f\rangle  
			= \langle \psi, S_\pm f\rangle 
			\,. \label{eqn-star2}
			\]
			The partial integration $(*)$ has no boundary terms 
			because $\supp(S'_\mp \psi)\cap \supp(S_\pm f)$
			is a subset of $J_\mp(\supp \psi) \cap J_\pm(
			\supp f)$, which is compact due to global
			hyperbolicity\footnote{%
				If $K,K'\subseteq M$ are two compact subsets
				of a globally hyperbolic $(M,g)$, then
				$J_+(K)\cap J_-(K')$ is compact
				\cite[A.5.4]{BaerGinPfae}.%
			}. 
	
			We found that $S_\pm = (S'_\mp)^*$. From
			$P^*S'_\mp=\Id_{\Gamma_0(\mathcal{E}^*)}$ we
			get, by forming the adjoint,
			$\Id_{\Gamma_0(\mathcal{E})} = (S'_\mp)^*P =
			S_\pm P$. Analogously one can deduce
			$\Id_{\Gamma_0(\mathcal{E}^*)} = S'_\pm P^*$.
	\end{enumerate}	
	So we find that $S_\pm$ are advanced/retarded Green's functions for
	$P$.  To show uniqueness let $T_\pm$ be another pair of advanced and
	retarded Green's functions for $P$. Equation (\ref{eqn-star2}) remains
	valid if we replace $S_\pm$ by $T_\pm$ (so that it reads $\langle
	S'_\mp\psi,f\rangle = \ldots = \langle\psi,T_\pm f\rangle$) and hence
	we obtain $T_\pm = S_\pm$.  
\end{proof}

\section{The Cauchy Problem}

In this section we will prove: 

\begin{theorem}
	\label{thm:cauchy-problem}
	Let $P\colon\Gamma(\mathcal{E})\to\Gamma(\mathcal{E})$ be a 
	prenormally hyperbolic first order linear differential operator on
	the vector bundle $\mathcal{E}$ on $(M,g)$ and let $\Sigma\subseteq M$
	be a smooth spacelike Cauchy hypersurface. Then the Cauchy problem 
	\[
	(P)\quad 
	\begin{cases}
		P\Phi = 0,\quad \Phi\in \Gamma(\mathcal{E}) \\
		\Phi|_\Sigma = \Phi_0
	\end{cases}
	\]
	has a unique solution for every initial datum $\Phi_0\in
	\Gamma_0(\mathcal{E}|_\Sigma)$. Moreover, this solution satisfies
	$\supp \Phi\subseteq J(\supp \Phi_0)$. 
\end{theorem}

\noindent Here, we denote by $\mathcal{E}|_\Sigma$ the restriction of the
bundle $\mathcal{E}$ on $M$ to the submanifold $\Sigma\subseteq M$. Moreover,
we set $J(K) := J_+(K) \cup J_-(K)$. 

Inspired by \cite{Dimock1982}, the idea of the proof is to reduce our first
order problem to a Cauchy problem for a second order normally hyperbolic
operator, and then we make use of the theory presented in \cite{BaerGinPfae}. 
Let $P,Q$ be a complementary pair of prenormally hyperbolic first oder
operators on 
$\mathcal{E}$ and let $\mathfrak{n}$ be the future directed unit normal vector
field along $\Sigma$. We define the following auxiliary Cauchy problems: 

\begin{tabular}{ll} 
	\begin{minipage}{70mm}
\begin{eqnarray*}
	(QP) \quad &&
\begin{cases}
	QP\Phi = 0, \quad \Phi\in \Gamma(\mathcal{E}) \\
	\Phi|_\Sigma = \Phi_0 \\
	(\nabla_{\mathfrak{n}}\Phi)|_\Sigma = \Psi_0
\end{cases} \\
&& \text{for given $\Phi_0,\Psi_0 \in \Gamma_0(\mathcal{E}|_\Sigma)$.}
\end{eqnarray*}
\end{minipage}&\begin{minipage}{70mm}
\begin{eqnarray*}
	(\widetilde{QP}) \quad &&
\begin{cases}
	QP\Phi = 0, \quad \Phi\in \Gamma(\mathcal{E}) \\
	\Phi|_\Sigma = \Phi_0 \\
	(P\Phi)|_\Sigma = 0
\end{cases} \\
&& \text{for given $\Phi_0 \in \Gamma_0(\mathcal{E}|_\Sigma)$.}
\end{eqnarray*}
\end{minipage}\\\begin{minipage}{70mm}
\begin{eqnarray*}
	(PQ) \quad &&
\begin{cases}
	PQ\Phi = 0, \quad \Phi\in \Gamma(\mathcal{E}) \\
	\Phi|_\Sigma = \Phi_0 \\
	(\nabla_{\mathfrak{n}}\Phi)|_\Sigma = \Psi_0
\end{cases} \\
&& \text{for given $\Phi_0,\Psi_0 \in \Gamma_0(\mathcal{E}|_\Sigma)$.}
\end{eqnarray*}
\end{minipage}&\begin{minipage}{70mm}
\begin{eqnarray*}
	(\widetilde{PQ}) \quad &&
\begin{cases}
	PQ\Phi = 0, \quad \Phi\in \Gamma(\mathcal{E}) \\
	\Phi|_\Sigma = \Phi_0 \\
	(Q\Phi)|_\Sigma = 0
\end{cases} \\
&& \text{for given $\Phi_0 \in \Gamma_0(\mathcal{E}|_\Sigma)$.}
\end{eqnarray*}
\end{minipage}
\end{tabular}

\begin{lemma}\label{lem:cauchy-probs}\mbox{ }\nopagebreak%
	\begin{enumerate}[(a)]
		\item
			If $\Phi\in\Gamma(\mathcal{E})$ solves $\widetilde{QP}$
			for initial datum $\Phi_0$, then $\Phi$ solves
			$(QP)$ for initial data $\Phi_0$ and $\Psi_0:=
			\nabla_\mathfrak{n}\Phi|_{\Sigma}$.
			The analogous result holds for $(\widetilde{PQ})$
			and $(PQ)$.

		\item \label{lem:cauchy-probs:b}
			For given $\Phi_0\in\Gamma_0(\mathcal{E}|_\Sigma)$
			there always exists a unique solution $\Phi\in
			\Gamma(\mathcal{E})$ to $(\widetilde{QP})$. It
			satisfies $\supp \Phi\subseteq J(\supp \Phi_0)$.
			The analogous statement holds for $(\widetilde{PQ})$. 

		\item \label{lem:cauchy-probs:c}
			$\Phi\in\Gamma(\mathcal{E})$ solves $(P)$ for
			initial datum $\Phi_0$ if and only if $\Phi$ solves
			$(\widetilde{QP})$ for initial datum $\Phi_0$. 
	\end{enumerate}
\end{lemma}

\begin{proof} 
	As a preparatory consideration, assume we are given 
	$\Phi\in\Gamma(\mathcal{E})$ and
	$\Phi_0\in\Gamma_0(\mathcal{E}|_\Sigma)$ such that 
	$\Phi|_\Sigma = \Phi_0$. Moreover, fix a point $x\in\Sigma$ and a
	local pseudo-orthogonal frame $(e_1,\ldots,e_n)$ on a neighborhood
	$U$, $x\in U\subseteq M$, such that $e_1|_{U\cap \Sigma}
	=\mathfrak{n}$ ($e_2,\ldots,e_n$ are then tangential to $\Sigma$ on
	$\Sigma\cap U$). Then we have the following equivalence of
	statements along $\Sigma\cap U$: 
	\begin{eqnarray}
	&& 0 = (P\Phi)|_\Sigma 
	= \sum_{i=1}^{n} \sigma_P(e_i^*) \nabla_{e_i}\Phi|_\Sigma +
	R_{P,U}\Phi|_\Sigma \nonumber\\
	\Leftrightarrow\ \quad &&
	\sigma_P(\mathfrak{n}^*) \nabla_\mathfrak{n}\Phi|_\Sigma
	\label{eqn-star1} 
	= -[\sigma_P(\mathfrak{n}^*)]^{-1} 
	\left( \sum_{i=2}^{n} \sigma_P(e_i^*) \nabla_{e_i}\Phi_0
	+ R_{P,U}\Phi_0 \right) \nonumber
	\end{eqnarray} 
	Here, $\sigma_P\in\Gamma(T^*M\otimes\mathrm{End}(\mathcal{E}))$
	denotes the principal symbol of $P$, $R_{P,U}\colon
	\Gamma(\mathcal{E}|_U)\to \Gamma(\mathcal{E}|_U)$ is a linear operator
	of order 0, and for $v\in T_xM$,
	$v^*\in T_x^*M$ denotes the co-vector associated with $v$ by the metric
	($v^* = g(v, \cdot)$, this is usual index-shifting).  We also 
	used invertibility of $\sigma_P$ in 
	normal direction, which follows from lemma
	\ref{lem:invertibility-sigma-P}. We now prove the three statements: 

	\begin{enumerate}[(a)]
		\item Let a solution $\Phi$ to $(\widetilde{QP})$ for initial
			datum $\Phi_0$ be given. We only have to verify that
			$\Psi_0 := \nabla_\mathfrak{n}\Phi|_\Sigma$ has indeed
			compact support; but this follows immediately from
			(\ref{eqn-star1}). 

		\item \textbf{Existence and support property:} 
			Let $\Phi_0$ be given. Inspired 
			by (\ref{eqn-star1}), we set using local
			pseudo-orthogonal frames 
			$(\mathfrak{n}, e_2, \ldots, e_n)$ along $\Sigma$:
			\[
			\Psi_0 := -[\sigma_p(\mathfrak{n}^*)]^{-1}
			\left( \sum_{i=2}^{n}\sigma_P(e_i^*)\nabla_{e_i}\Phi_0
			+ R_{P,U}\Phi_0 \right)
			\] 
			(so $\Psi_0 \in\Gamma_0(\mathcal{E}|_\Sigma)$). 
			Notice that the right hand side does not
			depend on the choice of local frame and hence $\Psi_0$
			is well defined globally along $\Sigma$. 
			Its support is a subset of $\supp(\Phi_0)$ and thus
			compact. 
			Using \onlinecite[thm.\ 3.2.11]{BaerGinPfae}, there
			is a solution $\Phi$ to $(QP)$ for initial datum
			$(\Phi_0,\Psi_0)$, which moreover satisfies the support
			property. Finally, the preparatory consideration
			at the beginning of this proof shows
			that our choice of $\Psi_0$ is such that the
			solution $\Phi$ satisfies $(P\Phi)|_\Sigma = 0$; 
			thus, $\Phi$ is a solution to $(\widetilde{QP})$ for
			initial datum $\Phi_0$. 

			\textbf{Uniqueness:} 
			If $\Phi'$ is another solution to $(\widetilde{QP})$ 
			then  by (a) and (\ref{eqn-star1}), it solves $(QP)$
			for $\Phi_0$ and $\Psi'_0 = \Psi_0$. Thus, $\Phi' =
			\Phi$ by uniqueness of solutions to $(QP)$
			(cf.\ \onlinecite[thm.\ 3.2.11]{BaerGinPfae}). 

		\item Fix $\Phi_0\in\Gamma_0(\mathcal{E}|_\Sigma)$. If
			$\Phi\in\Gamma(\mathcal{E})$ solves $(P)$ it is
			trivial that it also solves $(\widetilde{QP})$. 
			To prove the opposite direction, let
			$\Phi\in\Gamma(\mathcal{E})$ solve $(\widetilde{QP})$,
			i.\,e.\ $QP\Phi=0$, $\Phi|_\Sigma=\Phi_0$,
			$(P\Phi)|_\Sigma=0$. 
			By multiplying the first equation by $P$ from the
			left and by restricting the first equation to $\Sigma$
			we easily obtain 
			$PQP\Phi=0$, $(P\Phi)|_\Sigma = 0$, $(QP\Phi)|_\Sigma
			=0$, which means that $P\Phi$ solves
			$(\widetilde{PQ})$ for initial datum $\equiv0$.
			According to (\ref{lem:cauchy-probs:b}), solutions to 
			$(\widetilde{PQ})$ are unique, 
			hence, $P\Phi=0$ and $\Phi$ solves $(P)$.  \qedhere
	\end{enumerate}
\end{proof}

\begin{proof}[Proof of theorem \ref{thm:cauchy-problem}]
	Now this follows immediately from  
	\ref{lem:cauchy-probs}-\ref{lem:cauchy-probs:b} using
	\ref{lem:cauchy-probs}-\ref{lem:cauchy-probs:c}. 
\end{proof}

\noindent Finally, as a corollary we may write down the following
compatibility result for compactly supported Cauchy data on two different
Cauchy hypersurfaces $\Sigma$ and $\Sigma'$: 

\begin{corollary}[compatibility of Cauchy data] \label{cor:compatibility}
	Let $P\colon\Gamma(\mathcal{E})\to \Gamma(\mathcal{E})$ be as in theorem
	\ref{thm:cauchy-problem} and let $\Sigma,\Sigma'\subseteq M$ be two
	smooth spacelike Cauchy hypersurfaces. Then the spaces of compactly
	supported Cauchy data, $\Gamma_0(\mathcal{E}|_\Sigma)$ and
	$\Gamma_0(\mathcal{E}|_{\Sigma'})$ are in a one-one relation determined
	by having the same solution to the Cauchy problem: The map 
	\begin{eqnarray*}
		\Gamma_0(\mathcal{E}|_\Sigma) &&\to
		\Gamma_0(\mathcal{E}|_{\Sigma'}) \\ 
		\Phi_0 &&\mapsto \Phi|_{\Sigma'}
		\,, \ \text{for $\Phi$ the unique solution} \\
		&& \qquad\qquad \text{such that
		$\Phi|_\Sigma = \Phi_0$\,,}
	\end{eqnarray*}
	is an isomorphism of $\mathbbm{C}$-vector spaces. 
\end{corollary}

\begin{proof}
	This follows immediately by uniqueness of solutions to
	the Cauchy problem. To establish that $\Phi|_{\Sigma'}$ is compact,
	make use of the well known fact on globally hyperbolic spacetimes that
	for compact $K\subseteq M$, $J(K)\cap \Sigma'$ is compact. 
\end{proof}

\section{Application in CAR-algebraic Quantization}

\label{sec:quantisation-outline}

To illustrate an important application of our results, 
we shall briefly outline a basic 
quantization procedure using CAR-algebras for fermionic quantum fields
described by a partial differential equation $P\Phi = 0$  for
prenormally hyperbolic $P$ acting on sections of a vector bundle
$\mathcal{E}$ on $M$.  

Let $\Sigma\subseteq M$ be a smooth spacelike Cauchy hypersurface in $M$
with future directed unit normal vector field $\mathfrak{n}$ along $\Sigma$. 
We define the space of localized Cauchy data on $\Sigma$, 
\[
\mathscr{H}_\Sigma := \{ \Phi_0\in \Gamma_0(\mathcal{E}|_\Sigma)\}
\,,\]
and the space of solutions with localized Cauchy data,
\[
\mathscr{H} := \{ \Phi\in \Gamma(\mathcal{E})\,|\, P\Phi = 0\ \text{and}\
\Phi|_\Sigma\in \mathscr{H}_\Sigma \}
\,.\] 
Notice that due to the compatibility of compactly supported Cauchy data
(corollary \ref{cor:compatibility}), $\mathscr{H}$ is independent of the
choice of $\Sigma$. 


Now, assume there is a physically distinguished Hermitian inner product
$\langle \cdot, \cdot\rangle $ on $\mathscr{H}$ so that $(\mathscr{H}, \langle
\cdot, \cdot\rangle )$ forms a pre-Hilbert space.  Then the system could be
quantized by forming the CAR-algebra $\mathrm{CAR}(\mathscr{H}, \langle \cdot,
\cdot\rangle )$ (CAR = canonical anti-commutation relations, for definition of
the unique CAR-algebra generated by the elements of a pre-Hilbert space cf.\
\onlinecite{BraRob2}). 

To perform such a construction, it may often be easier to work on a Cauchy
hypersurface. For every smooth spacelike Cauchy hypersurface
$\Sigma$, the canonical $\mathbbm C$-vector space isomorphism 
\[
\mathscr{H} \to \mathscr{H}_\Sigma \,, \qquad 
\Phi \mapsto \Phi|_\Sigma  \,,
\]
(with inverse given by assigning to $\Phi_0$ the solution of the Cauchy
problem with initial datum $\Phi_0$, cf.\ theorem \ref{thm:cauchy-problem})
can be used to pull back a Hermitian inner product $\beta_\Sigma$ on
$\mathscr{H}_\Sigma$ to a Hermitian inner product $\beta$ on $\mathscr{H}$, so
that an isometric isomorphism $(\mathscr{H}, \beta)\cong (\mathscr{H}_\Sigma,
\beta_\Sigma)$ would be obtained. 
Of course, such a construction will only be satisfactory if $\beta$ is 
independent of the choice of Cauchy hypersurface $\Sigma$. 


\cite{Dimock1982} presents such a construction for the spin $\frac{1}{2}$
Dirac equation 
\[
(\gamma_a\nabla^a + m)\, \Phi = 0
\,,\qquad \Phi\in \Gamma(\mathcal{D}M)
\,,\] 
where $\mathcal{D}M$ denotes the bundle of Dirac spinors on $M$. 
Here, the principal symbol $\gamma_a\in T^*M\otimes
\mathrm{End}(\mathcal{D}M)$ of the Dirac operator $D = \gamma_a\nabla^a$ is
given by the Dirac matrices $\gamma_a$. 
The Hermitian product $\beta_\Sigma$ on $\mathscr{H}_\Sigma$ is given by 
\[
\beta_\Sigma(\Psi,\Phi) 
:= \int\limits_\Sigma \mathfrak{n}_a\,j^a(\Psi,\Phi)\,d\mu_\Sigma
\,,\] 
where $j_a(\Psi,\Phi) = \langle \Psi^+, \gamma_{a}(\Phi) \rangle$ is the
Dirac current (depending sesquilinearly on two Dirac spinor fields
$\Psi,\Phi\in\Gamma(\mathcal{D}M)$), 
$\Phi^+\in \Gamma(\mathcal{D}^*M)$ is the Dirac adjoint co-spinor field
of $\Phi$ (which depends complex anti-linearly on $\Phi$), 
$\langle \cdot, \cdot\rangle \colon
\Gamma(\mathcal{D}^*M)\times\Gamma(\mathcal{D}M)\to C^\infty(M;\mathbbm{C})$
denotes the canonical fiberwise pairing of co-spinor fields and spinor fields
and $d\mu_\Sigma$ is the metric-induced volume density on $\Sigma$. 

As can be shown, $\beta_\Sigma$ is indeed a Hermitian scalar product on
$\mathscr{H}_\Sigma$ (i.\,e.\ it is positive definite), and it does not
depend on the choice of Cauchy hypersurface (which is basically a consequence
of Stokes' theorem), in the sense that if for a
second smooth spacelike Cauchy hypersurface $\Sigma'$ we pull back
$\beta_{\Sigma'}$ from $\mathscr{H}_{\Sigma'}$ to $\mathscr{H}$, the
resulting Hermitian scalar product $\beta$ on $\mathscr{H}$ equals the one
obtained by pulling back $\beta_\Sigma$ from $\mathscr{H}_\Sigma$. In fact, we
thus have a whole chain of isometric isomorphisms of pre-Hilbert spaces, 
\[
(\mathscr{H}, \beta)
\cong (\mathscr{H}_\Sigma, \beta_\Sigma) 
\cong (\mathscr{H}_{\Sigma'}, \beta_{\Sigma'})
\cong \ldots
\,,\]
resulting in a chain of canonically isomorphic CAR-algebras 
\[
\mathrm{CAR}(\mathscr{H}, \beta)
\cong \mathrm{CAR}(\mathscr{H}_\Sigma, \beta_\Sigma) 
\cong \mathrm{CAR}(\mathscr{H}_{\Sigma'}, \beta_{\Sigma'})
\cong \ldots
\,.\] 

Imitating this basic construction for fields of higher spin is a challenging
task.  Generally, one would try to construct the Hermitian inner product
$\beta_\Sigma$ for a Cauchy hypersurface $\Sigma$ within a setting where the
field equation is given by a Lagrangian density and by making use of a
suitable concept of complex conjugation of fields. E.\,g.\ for the higher spin
field equations proposed by \cite{Buchdahl1982a} and reformulated by
\cite{Wuensch1985}, such a Lagrangian formulation was presented by
\cite{Illge1993}.  However, as was shown in \cite{Dipl}, the resulting product
$\beta_\Sigma$ fails to be positive definite (and thus it is not an Hermitian
inner product) as soon as spin is $>1$. 

This state of affairs motivates future research, illuminating in which cases
of higher spin prenormally hyperbolic partial differential equations a
quantum field theory based on a canonical choice of CAR- (or CCR-) algebra can
be constructed.

\begin{acknowledgments}
	These results go back to work under supervision of Prof.\ Rainer
	Verch, Leipzig, in 2007, cf.\ \cite{Dipl}. I would like to thank him
	as well as Prof.\ Christian B\"{a}r, Potsdam, for their support. 
\end{acknowledgments}

\bibliography{cauchy-problem.bib}

\end{document}